\documentclass[aip,jmp,12pt]{revtex4-1}

\usepackage{epsfig,graphics,graphicx}
\usepackage{dcolumn}
\usepackage{bm}
\usepackage{amsmath,amssymb,amsthm}
\usepackage{centernot}
\usepackage{bbold}
\usepackage{soul}
\usepackage{mathtools}
\usepackage{fullpage}
\usepackage{units}
\usepackage[usenames,dvipsnames]{xcolor}
\usepackage{braket}
\usepackage{indentfirst}
\usepackage{pgf,tikz}
\usepackage{mathrsfs}
\usetikzlibrary{arrows}
\usepackage{lipsum}
\usepackage[export]{adjustbox}

\usepackage[utf8]{inputenc}
\usepackage[T1]{fontenc}

\usepackage{dsfont}

\usepackage{mathptmx}
\usepackage{amssymb}
\usepackage{amsmath}
\usepackage{latexsym}
\usepackage{amsfonts}

\usepackage{hyperref}
\hypersetup{pdfpagemode=UseNone}

\newtheorem{theorem}{Theorem}
\newtheorem{definition}[theorem]{Definition}

\newtheorem{lemma}[theorem]{Lemma}
\newtheorem{proposition}[theorem]{Proposition}
\newtheorem{corollary}[theorem]{Corollary}

\newcounter{rem}
\newcommand{\mc}[1]{\mathcal{#1}}

\newcommand{\td}[1]{\tilde{#1}}

\def\>{\rangle}
\def\<{\langle}

\renewcommand{\rho}{\varrho}

\def\textbf#1{{\bf #1}}

\newcommand{\pandms}{prepare-and-measure statistics }
\newcommand{\Tppm}{\sum_{k,j,i}q_{O}(\tilde{k} \vert k)p(k \vert 
j,i)q_{P}(i \vert \tilde{i} \,\, )q_{M} (j \vert \tilde{j} \,\,)}
\newcommand{\behav}[3]{p({#1} \vert {#2},{#3})}
\newcommand{\behavtd}[3]{p(\tilde{#1} \vert \tilde{#2},\tilde{#3})}

\def\beq{\begin{equation}}
\def\eeq{\end{equation}}
\def\beqa{\begin{eqnarray}}
\def\eeqa{\end{eqnarray}}
\def\eea{\end{array}}
\def\bea{\begin{array}}
\newcommand{\bei}{\begin{itemize}}
\newcommand{\eei}{\end{itemize}}
\newcommand{\bee}{\begin{enumerate}}
\newcommand{\eee}{\end{enumerate}}
\def\bep{\begin{proposition}}
\def\eep{\end{proposition}}
\def\bel{\begin{lemma}}
\def\eel{\end{lemma}}
\def\bet{\begin{theorem}}
\def\eet{\end{theorem}}
\def\bed{\begin{definition}}
\def\eed{\end{definition}}

\definecolor{cgreen}{RGB}{26, 199, 76}

\usepackage{soul}
\definecolor{violeta}{cmyk}{0.07,0.90,0,0.34}

\begin{document}
\title{Resource theory of contextuality for arbitrary 
prepare-and-measure experiments}

\author{Cristhiano Duarte} 
\affiliation{International Institute of Physics, Federal University of Rio 
Grande do Norte, 59078-970, P. O. Box 1613, Natal, Brazil}

\author{Barbara Amaral}
\affiliation{Departamento de F\'isica e Matem\'atica, CAP - Universidade 
Federal de S\~ao Jo\~ao del-Rei, 36.420-000, Ouro Branco, MG, Brazil} 
\affiliation{International Institute of Physics, Federal University of Rio 
Grande do Norte, 59078-970, P. O. Box 1613, Natal, Brazil}

\date{\today}


\begin{abstract}
Contextuality has been identified as  a potential 
resource responsible for the quantum advantage in several tasks. It is then necessary to develop a resource-theoretic 
framework for contextuality, both in its standard and generalized forms. Here we provide 
 a formal resource-theoretic approach for generalized
contextuality based on a physically motivated set of free operations with an explicit parametrization.
Then, using  an efficient linear 
programming characterization for the noncontextual set 
of prepared-and-measured statistics, we adapt known resource quantifiers for contextuality and nonlocality  to obtain natural monotones for 
generalized contextuality in arbitrary prepare-and-measure experiments.
\end{abstract}

\pacs{03.65.Ta, 03.65.Ud, 02.10.Ox}

\maketitle



\section{Introduction}

\maketitle

Prepare-and-measure experiments provide simple situations in which the differences between classical and nonclassical 
probabilistic theories can be explored. One such difference is related to the generalized notion of noncontextuality, a condition imposed
on ontological models that asserts 
that operationally indistinguishable laboratory operations  should be represented identically in the model \cite{RS05,SSW17,MPKRS16,RK17}. 
Inconsistencies between observed data and the existence of such a model can be understood as a signature of nonclassicality.


Besides its importance for foundations of physics \cite{Specker60,KS67,RS05}, noncontextuality has been identified  as a potential 
resource responsible for the quantum advantage in several tasks \cite{VFGE12,Raussendorf13,UZZWYDDK13, HWVE14, DGBR14, BDBOR17, SS17, SHP17}. Hence,
it is important  to investigate contextuality in arbitrary prepare-and-measure experiments from the perspective of resource theories,
which give
 powerful frameworks for the formal treatment of a 
physical property as an operational resource 
\cite{SOT16, 
BG15,CMH16, TF17, CFS16,ACTA17,GL15,BHORS13}.

It is commonly understood, see Refs.~\cite{BG15,ACTA17,TF17,CFS16} for 
instance, 
that such theories consist in the specification of three main components: 
$i)$ a class $\mc{O}$ of \emph{objects}, that represent those entities one 
aims to manipulate seeking for some gain or benefit, and that may possess 
the resource under consideration;
$ii)$ a special class $\mc{F}$ of transformations, called the \emph{free 
operations}, that fulfills the essential requirement of mapping every 
 resourceless object of the theory into another resourceles object, 
\emph{i.e.} a set of transformations that does not create a new 
resource from a resourceless object; and $iii)$ a \emph{measure} or a
\emph{quantifier} that 
outputs the amount of resource a given object contains. For consistency, 
the fundamental requirement for a function to be a valid quantifier is that 
of being a monotone with respect to the considered 
resource: every quantifier is non-increasing under the corresponding free 
operations.

Resource-theoretic approaches for quantum nonlocality are highly 
developed \cite{Barrett05b, Allcock09, GWAN12, Joshi13, Vicente14, LVN14, GA15, 
GA17} and  the operational 
framework of the  standard notion of contextuality as a resource has received much attention lately \cite{HGJKL15, GHHHHJKW14, ACTA17, ABM17}.
Nonetheless, a proper treatment for
 the generalized framework of prepare-and-measure experiments 
considered in Refs.~\cite{RS05,SSW17,MPKRS16,RK17} as a resource is still missing. Here, using the novel 
generalized-noncontextual polytope, an efficient linear 
programming~\cite{LexSchrivjer99} characterization for the contextual set 
of prepared-and-measured statistics presented in Ref.~\cite{SSW17}, we 
present a mathematically well structured resource-theoretic approach for 
generalized contextuality based on a physically motivated set of free operations with an explicit parametrization.
We then adapt known resource quantifiers for contextuality and nonlocality \cite{Barrett05b, Allcock09, GWAN12,
Joshi13, Vicente14, LVN14,  GHHHHJKW14,
GA15, HGJKL15,  GA17, ACTA17, BAC17, ABM17} to obtain natural monotones for 
generalized contextuality in arbitrary prepare-and-measure experiments.

This work is organized as follows: in Sec.~\ref{sec: GenContextuality} we 
review the definition of generalized non-contextuality and 
the linear programming characterization of the noncontextual set; in 
Sec.~\ref{sec:RTofContextuality} we introduce the three important 
components of the resource theory: in Subsec.~\ref{subsec:Objects} we 
define the objects of the theory, in Subsec.~\ref{subsec:FreeOperations} we 
provide a set of physically motivated free operations for generalized 
contextuality in prepare-and-measure experiments, and in 
Subsec.~\ref{subsec:Quantifiers} we list several contextuality
quantifiers and we explicitly prove that they are monotones with respect to the set of free operations defined in Subsec.~\ref{subsec:FreeOperations}; we finish with
discussion and open questions in Sec.~\ref{sec:conc}.

\section{Generalized Contextuality}
\label{sec: GenContextuality}

\subsection{A glimpse on the theory}
\label{sec:glimpse}
We consider a prepare-and-measure experiment with a set of possible  
preparations $\mathcal{P}=\left\{P_1,P_2, \cdots , P_I\right\}$, 
a set of possible measurements $\mathcal{M}=\left\{M_1,M_2, \cdots , 
M_J\right\}$, each measurement with possible outcomes
$\mathcal{D}=\left\{d_1, d_2, \ldots, d_K\right\}$.  An operational  
probabilistic theory that describe 
this prepare-and-measure experiment 
specifies, for  each measurement $M_j$, a probability distribution $p(k\vert j,i)$ over $\mathcal{D}$
wich specifies the probability of obtaining outcome $d_k$ when performing measurement $M_j$,
conditioned on the preparation $P_i$. We denote the measurement event of measuring $M_j$ and obtaining outcome $d_k$ as
$k \vert j$.

\begin{definition}
 Two preparations $P_i$ and $P_{i'}$ are \emph{operationally equivalent} if
 \beq p(k\vert j,i)= p(k\vert j,i') \ \forall d_k\in \mathcal{D}, M_j \in \mathcal{M}.\eeq
\end{definition}

In other words, $P_i$ and $P_{i'}$ are said to be operationally equivalent if they give the same statistics for every measurement.
Operational equivalence between $P_i$ and $P_{i'}$ will be denoted by
$P_i \simeq P_{i'}$.

\begin{definition}
 Two measurement events $k\vert j$ and $k'\vert j'$ are \emph{operationally equivalent} if  
 \beq p(k\vert j,i)= p(k'\vert j',i) \ \forall P_i \in \mathcal{P}.\eeq
\end{definition}

In other words, $k\vert j$ and $k'\vert j'$ are said to be operationally 
equivalent whenever they have the same statistics for every preparation 
in $\mathcal{P}$. 
Operational equivalence between $k\vert j$ and $k' \vert j'$ will be denoted by
$k\vert j \simeq k'\vert j'$.

We then specify a set  $\mathcal{E}_P$ of operational equivalences for the preparations
\beq \sum_i \alpha_i^sP_i \simeq \sum_i \beta_i^sP_i, \ s=1, \ldots , \left|\mathcal{E}_P\right|\eeq
where $\sum_i \alpha_i^sP_i$ and $\sum_i \beta_i^sP_i$ represent convex combinations of the preparations $P_i$, and 
a set $\mathcal{E}_M$ of operational equivalences for the measurement effects
\beq \sum_{k,j} \alpha_{k\vert j}^r\left[k\vert j\right] \simeq \sum_{k,j} \beta_{k\vert j}^r\left[k\vert j\right],
\ r=1, \ldots , \left|\mathcal{E}_M\right|\eeq
where $\sum_{k,j} \alpha_{k\vert j}^r\left[k\vert j\right]$ and $\sum_{k,j} \beta_{k\vert j}^r\left[k\vert j\right]$ represent convex combinations of measurement events.

\begin{definition}
 A \emph{prepare-and measure scenario} \beq\mathcal{S}\coloneqq\left\{\mathcal{P}, \mathcal{M}, \mathcal{D}, \mathcal{E}_P, \mathcal{E}_M\right\}\eeq consists of a set of preparations
 $\mathcal{P}$, a set of measurements $\mathcal{M}$, a set of outcomes $\mathcal{D}$, a set of operational equivalences for 
 the preparations  $\mathcal{E}_P$ and a set of operational equivalences for the measurements $\mathcal{E}_M$.
A prepare-and-measure statistics (more 
commonly known as behaviours or 
black-box correlations~\cite{BCPSW13,NGHA15,Slofstra17}) is a set of conditional probability distributions
\begin{equation}
\label{def:behaviour}
\boldsymbol{B} \coloneqq \left\{p\left(k \vert 
j,i\right)\right\}_{j \in [J], i \in [I], k \in [K]}
\end{equation}
that give the probability of outcome $d_k$ for each measurement $M_j$ 
given the preparation $P_i$.
\end{definition}

A schematic representation of a prepare-and-measure scenario is shown in Fig. \ref{Figure_prepare_and_measure}.

\begin{figure}[]
\includegraphics[scale=0.1]{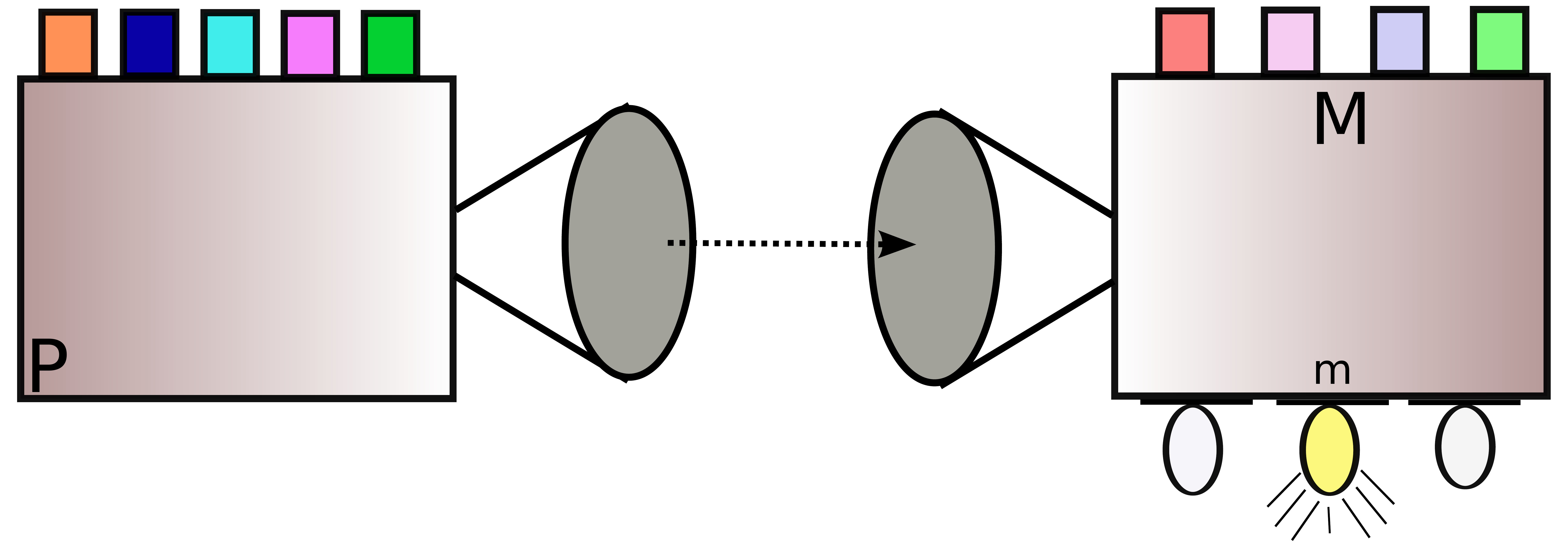}
 \caption{\footnotesize{Schematic drawing of a prepare-and-measure 
scenario $\mc{S}$. Each button out of $\vert I \vert$ above the box 
labeled with a $P$ represents
a possible preparation outputted for the box. Analogously each button out 
of $\vert J \vert$ presented above the box labeled with an $M$ represent 
a choice of measurement. The lamp bulbs below that box mean a possible 
outcome for each chosen and pressed measurement button. Together with 
these boxes  are given operational equivalences $\mc{E}_{M}$ and 
$\mc{E}_{P}$. It is the whole structure $\mc{S}=\{\{P_i\}_{i}, 
\{M_j\}_{j}, \{m_k\}_{k},\mc{E}_{M},\mc{E}_{P} \}$ which we call a 
prepare-and-measure scenario.} \label{Figure_prepare_and_measure}}
\end{figure}

\subsubsection{Ontological models}

\begin{definition}
 An \emph{ontological model} for  a prepare-and-measure statistics $B=\left\{p\left(k\vert j,i\right)\right\}$ is a specification of a set of \emph{ontic states} $\Lambda$, for each preparation $P_i$ a probability space
 $\left(\Lambda, \Sigma, \mu_i\right)$ and for each $\lambda \in \Lambda$ 
and each $M_j \in \mathcal{M}$ a probability distribution  
$\left\{\xi_{k\vert j}\left(\lambda\right)\right\}$ over 
 $\mathcal{D}$, such that  \beq  p\left(k\vert j,i\right) = \int_{ 
\Lambda} \xi_{k\vert 
j}\left(\lambda\right)\mu_i\left(\lambda\right).\label{eq:model}\eeq
\end{definition}

\begin{figure}[]
\includegraphics[scale=0.15]{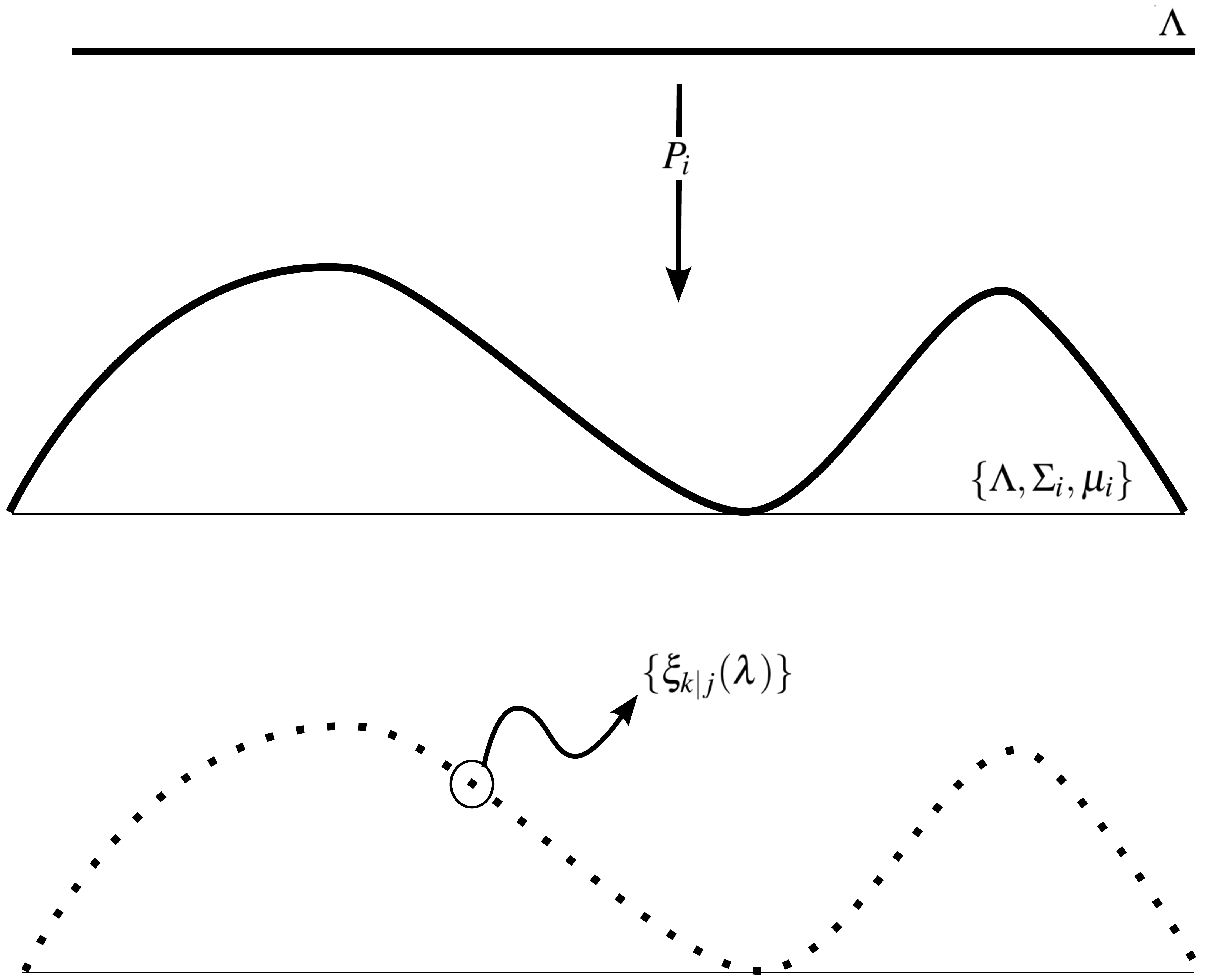}
 \caption{\footnotesize{Schematic drawing of the 
geometrical/probabilistic meaning of an ontological model as proposed by 
the authors in Refs.~\cite{MPKRS16,RS05}. Each preparation $P_i$ 
determines a probability space $(\Lambda,\Sigma,\mu_{P_i})$ whose 
underlying set is $\Lambda$, with $\Sigma$ and $\mu_{P_{i}}$ being a 
$\sigma$-algebra in $\mc{P}(\Lambda)$ and a probability measure over 
$\Lambda$, respectively. Roughly speaking, 
it means that some regions on $\Lambda$ are more likely than others. Now, 
for each set of measurement events $\{m_{k} \vert M_{j}\}$, and for each 
ontological state $\lambda \in \Lambda$ there is associate with them a 
collection of response functions $\{\xi_{1\vert M_1}(\lambda),\xi_{2\vert 
M_1}(\lambda), ..., \xi_{d\vert M_1}(\lambda),...,\xi_{d\vert 
M_{\vert J \vert}}(\lambda)\}$, determining the output of the measurement 
procedure when it is described by the ontological state $\lambda$.} 
\label{Figure_geometrical_meaning}}
\end{figure}

The interpretation of an ontological model is shown in Fig. 
\ref{Figure_geometrical_meaning}.
The ontic state $\lambda$ is understood as a variable that describes the behavior of the system that may not be accessible experimentally. 
If preparation $P_i$ is implemented, the 
ontic state $\lambda$ is sampled according to the associated probability 
distribution $\mu_{i}$. On the other hand, given $\lambda$, for every 
measurement $M_j$ the outcome $d_k$ is  a probabilistic function of 
$\lambda$, described by the response functions $\xi_{k\vert 
j}\left(\lambda\right)$. The variable $\lambda$ mediates the correlations 
between measurements and preparations. 
From the perspective of causal 
models \cite{Pearl00}, Eq. \eqref{eq:model} implies that the prepare-and-measure 
statistics is consistent with the causal structure shown in Fig. 
\ref{fig:causal}.

\begin{figure}[h]
 \definecolor{bcduew}{rgb}{0.7372549019607844,0.8313725490196079,0.9019607843137255}
 \begin{tikzpicture}[line cap=round,line join=round,>=triangle 45,x=1cm,y=1cm]
 \draw [->,line width=1pt] (4,1) -- (4,-0.65);
 \draw [->,line width=1pt] (4,-1) -- (6.65,-1);
 \draw [->,line width=1pt] (7,1) -- (7,-0.65);
 \draw [fill=bcduew] (4,1) circle (10pt);
 \draw[color=black] (4,1) node {$i$};
 \draw [fill=bcduew] (7,1) circle (10pt);
 \draw[color=black] (7,1) node {$j$};
 \draw [fill=bcduew] (4,-1) circle (10pt);
 \draw[color=black] (4,-1) node {$\lambda$};
 \draw [fill=bcduew] (7,-1) circle (10pt);
 \draw[color=black] (7,-1) node {$k$};
\end{tikzpicture}
\caption{Causal structure of an ontological model for a prepare-and-measure experiment. Given a preparation $i$, the ontic state $\lambda$ is sampled according to
$\mu_i$. then, for a choice of measurement $j$, the value of $k$ is sampled according to $\xi_{k\vert j}\left(\lambda\right)$. Notice that the variable $\lambda$ mediates the correlations 
between measurements and preparations. }
\label{fig:causal}
\end{figure}
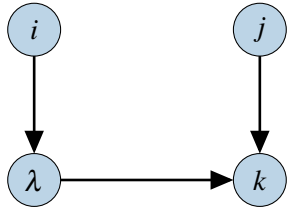

\subsubsection{Noncontextual models}

The generalized notion of noncontextuality introduced in Ref. \cite{RS05} requires that preparations and measurement events that
can not be distinguished operationally are identically represented in the model.
This implies that  the operational equivalences valid for $\mathcal{P}$ and $\mathcal{M}$ should also be valid for the functions 
$\mu_i$ and $\xi_{k\vert j}$, respectively. In terms of the previous definitions, we have:

\begin{definition}[Noncontextual ontological models]
  An ontological model satisfies \emph{preparation noncontextuality} if $\mu_{i}=\mu_{i'}$ whenever $P_i$ and $P_{i'}$ are operationally
  equivalent. 
 An ontological model satisfies \emph{measurement noncontextuality} if 
$\xi_{k\vert j}=\xi_{k' \vert j'}$ whenever $k\vert j$ and $k'\vert j'$
 are operationally equivalent. 
 An ontological model is \emph{universally noncontextual}, or simply \emph{noncontextual}, if it satisfies both preparation and measurement noncontextuality.
\end{definition}

The non-existence of a noncontextual ontological model for the prepare-and-measure statistics $B$ can be interpreted as
signature of the nonclassicality of $B$. It is a known fact that some
prepare-and-measure statistics obtained with quantum systems do not admit a noncontextual ontological model \cite{RS05}.

\begin{definition}
 The prepare-and-measure statistics $B$ is called \emph{noncontextual} if it has an noncontextual ontological model. The set of all noncontextual prepare-and-measure statistics for the scenario $\mathcal{S}$
 will be denoted by 
 $\mathsf{NC}\left(\mathcal{S}\right).$
\end{definition}

\subsection{Linear Characterization}
\label{sec:linear_characterization}

It was shown in Ref. \cite{SSW17} that if a prepare-and-measure statistics $B$ has a  noncontextual ontological model, 
then it also has a noncontextual  ontological model with an ontic state space $\Lambda$ of
finite cardinality. This implies that membership in $\mathsf{NC}\left(\mathcal{S}\right)$ can be formulated in terms of linear programming.

 Given a ontic state $\lambda$, the value of the response functions $\xi_{k\vert j}$ can be represented in a vector 
 \beq \boldsymbol{\xi}\left(\lambda\right)\coloneqq \left(\xi_{1\vert 1}\left(\lambda\right), \ldots,
 \xi_{K\vert 1}\left(\lambda\right), \ldots \xi_{1\vert J}\left(\lambda\right), \ldots, \xi_{K\vert J}\left(\lambda\right)
 \right)\eeq

\begin{definition}
 For fixed $\lambda \in \Lambda$, the vectors $\boldsymbol{\xi}\left(\lambda\right)$ defined by different choices of response functions $\xi_{m\vert M}$ satisfying 
 measurement noncontextuality are called  \emph{noncontextual measurement 
assignments}. The set of all noncontextual measurement assignments is a 
called the \emph{noncontextual measurement-assignment polytope}.
\end{definition}

As shown in Ref. \cite{SSW17}, fixed $\lambda \in \Lambda$, the set of all noncontextual measurement 
assignments is indeed a polytope since it is characterized by the 
following linear restrictions:
\begin{align}
 \xi_{k\vert j}\left(\lambda\right)&\geq 0\\
 \sum_k \xi_{k\vert j}\left(\lambda\right)&=1\\
 \sum_{k,j}\left(\alpha_{k\vert j}^r -\beta_{k\vert j}^r\right)\xi_{k\vert 
j}\left(\lambda\right)&=0.
\end{align}
Notice that since  these constraints do not depend on $\lambda$, the noncontextual measurement-assignment 
polytope is the same for every $\lambda$. We denote by 
$\tilde{\boldsymbol{\xi}}\left(\kappa\right)$ the 
extremal points of this polytope, with $\kappa$  a discrete variable ranging over some enumeration of these extremal points.

\begin{proposition}
A prepare-and-measure statistics $B=\left\{p\left(k\vert 
j,i\right)\right\}$ in the scenario $\mathcal{S}$ has a noncontextual 
ontological model if, and only if, there is a set of probability 
distributions $\left\{\mu_i\left(\kappa\right)\right\}$
over $\kappa$ such that 
\begin{align}
\sum_i\left(\alpha^s_i - \beta^s_i\right)\mu_i\left(\kappa\right)&=0     \\
\sum_{\kappa} \tilde{\xi}_{k\vert j}\left(\kappa\right)\mu_i\left(\kappa\right)&=p\left(k\vert j,i\right),
\end{align}
where $\kappa$ ranges over the discrete set of vertices of the measurement-assignment polytope.

\end{proposition}

This proposition implies that membership in $\mathsf{NC}\left(\mathcal{S}\right)$ can be efficiently tested using linear programming, which in 
turn implies that some of the quantifiers proposed in Sec.\ref{subsec:Quantifiers} can also be computed efficiently using linear programming.

\section{Resource Theory of Generalized Contextuality}
\label{sec:RTofContextuality}

\subsection{Objects}
\label{subsec:Objects}

We define the set $\mc{O}$ of \emph{objects} as 
the collection of all allowed prepare-and-measure statistics:
\begin{equation}
\boldsymbol{B} = \left\{p\left(k \vert 
j,i\right)\right\}_{j \in J, i \in I, k \in K}
\end{equation}
for a prepare-and-measure scenario $\mathcal{S}$. The \emph{free objects}, or 
\emph{resourceles} ones, correspond to those 
behaviors $\boldsymbol{B}\in \mathsf{NC}\left(\mathcal{S}\right)$ with a universally 
noncontextual model.

Given two objects $\boldsymbol{B}_{1}$ and $\boldsymbol{B}_{2}$ we also allow for
a combination of them in order to obtain a third new object, denoted by 
$\boldsymbol{B}_{1} \otimes \boldsymbol{B}_{2}$. One may think of this such an object 
$\boldsymbol{B}_{1} \otimes \boldsymbol{B}_{2}$ as representing full access to both 
$\boldsymbol{B}_{1}$ and $\boldsymbol{B}_{2}$ together and at the same time. For our 
purposes it will be enough to consider that combination as the juxtaposition of two 
independent behaviours:

\begin{definition}
 Given two behavior  
$B_1$ and $B_2$ (not necessarily in the same scenario), the juxtaposition of $B_1$ and 
$B_2$, denoted by $B_1 \otimes B_2$,
is the behavior obtained by independently choosing  preparation and measurement  for 
$B_1$ and $B_2$. That is,
the preparations in $B_1\otimes B_2$ correspond to a pair of preparations, $i_1$ for 
$B_1$ and 
$i_2$ for $B_2$ and analogously for the measurements. The corresponding probability 
distributions are given by
\beq p\left(k_1k_2\vert j_1j_2, i_1i_2\right)=p\left(k_1\vert j_1, 
i_1\right)p\left(k_2\vert j_2, i_2\right).\eeq
\end{definition}

As expected, the juxtaposition of two noncontextual behaviors is a noncontextual 
behavior. 

\begin{theorem}
If $B_1$ and $B_2$ are noncontextual if and only if $B_1 \otimes B_2$ is noncontextual.
\end{theorem}

\begin{proof}
 Let $\left(\Lambda_1, \Sigma_1, \mu_{i_1}\right)$ and $\left\{\xi_{k_1\vert j_1} 
\right\}$ be an ontological model for $B_1$ and 
 $\left(\Lambda_2, \Sigma_2, \mu_{i_2}\right)$ and $\left\{\xi_{k_2\vert j_2}\right\}$ 
be an ontological model for $B_2$.
 Then $\left(\Lambda_1 \times \Lambda_2, \Sigma_1 \times \Sigma_2, \mu_{i_1}\times 
\mu_{i_2}\right)$ and $\left\{\xi_{k_1\vert j_1}\times \xi_{k_2\vert j_2} \right\}$
 is an ontological model for $B_1 \otimes B_2$. Conversely, if an ontological model for 
$B_1 \otimes B_2$ is given, an ontological model for 
 $B_1$ can be obtained by marginalizing over $B_2$ and an ontological model for 
 $B_2$ can be obtained by marginalizing over $B_1$.
\end{proof}

\subsection{Free Operations}
\label{subsec:FreeOperations}

Given a prepare-and-measure scenario $\mc{S}$, we define the 
set $\mc{F}$ of free operations in analogy with simulation 
of communication channels~\cite{TF17}: the image 
$\boldsymbol{\tilde{B}}=T(\boldsymbol{B})$ of $B$ through every mapping $T:\mc{O} 
\longrightarrow \mc{O}$ in $\mc{F}$ should be viewed as a simulation of a
new scenario using one preprocessing for the preparation box $\mathcal{P}$, another 
preprocessing for the measurement box $\mathcal{M}$, and for the last a 
post-processing for the outcomes $\mathcal{D}$ of each measurement $M_j$ (see 
Fig.\ref{Figure_free_operations_boxes}). More formally, each free 
operation is a map:
\begin{align}
\label{def:free_operations}
T: \mc{O} &\longrightarrow \mc{O}  \\
          \{p(k \vert j,i)\} &\mapsto \{p({\tilde{k}} \vert 
\tilde{j},\tilde{i}) \} \nonumber
\end{align}
where for each $\td{k} \in \td{K},\td{i} \in \td{I},\td{j} \in \td{J}$:
\begin{equation}
\label{def:free_operations2}
 \sum_{k,j,i}q_{O}({\tilde{k}} \vert {k})p(k \vert j,i)q_{P}(i 
\vert {\td{i}})q_{M}(j \vert 
{\td{j}}),
\end{equation}
with $q_{P}:\tilde{I} \longrightarrow I$, 
$q_{M}:\tilde{J} \longrightarrow J$, and $q_O: 
K \longrightarrow \tilde{K}$ stochastic maps~\cite{NC00} from certain 
input alphabets to another sets of output alphabets. In what follows, the stochastic map $q_O$ can also depend on the measurement $\tilde{j}$, that is, different  post-processing of the outcomes can be applied to different measurements.
Hence, it would be more appropriate to write $q^{\tilde{j}}_O$, but we avoid the use of this heavy notation.
Eq.~\eqref{def:free_operations2} shows that, after a 
suitable relabeling of the indexes, the overall effect of each free 
operation is a right-multiplication of a stochastic matrix and a 
left-multiplication of another stochastic matrix on prepare-and-measure 
statistics $\{p(k \vert j,i)\}$, thus each $T$ in $\mc{F}$ acts as a 
linear map on the set of objects. We have proved, therefore, the following 
results:
\begin{figure}[]
\includegraphics[scale=0.1]{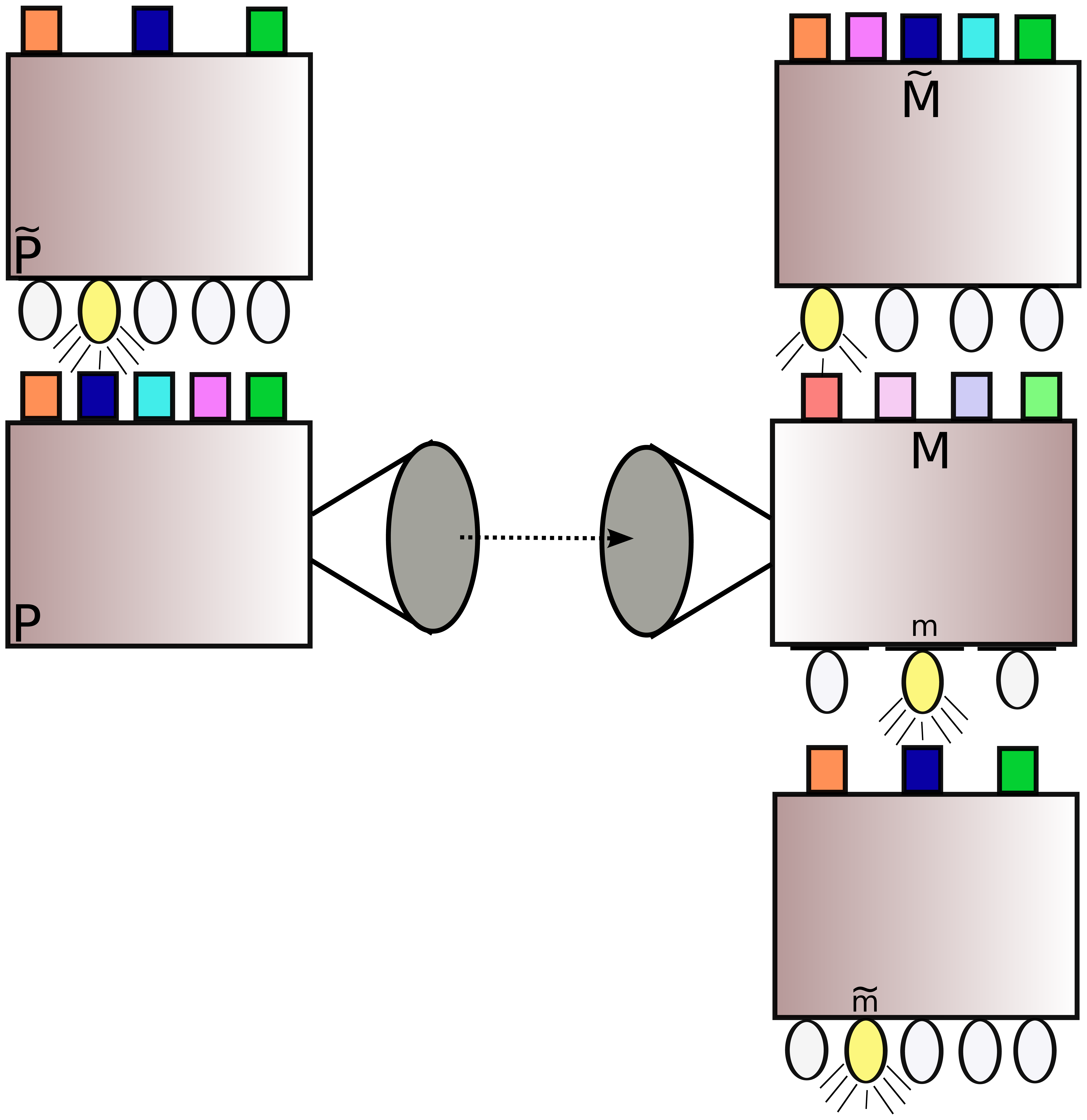}
 \caption{\footnotesize{Schematic drawing of general free operation $T$ in 
$\mc{F}$. The image of $T$ on a particular scheme $\boldsymbol{B}$ is 
viewed as a simulation of a new scenario $\boldsymbol{\tilde{B}}$, 
through two preprocessing and one post-processing. In the communication 
theory parlance, we are applying two steps of encoding, one 
$enc_{P}:\tilde{I} \longrightarrow I$ in the preparation step and another 
$enc_{M}:\tilde{J} \longrightarrow J$ in the measurement step, followed by 
an extra step of decoding $dec: [d] \longrightarrow [\tilde{d}]$. The net 
result of such a procedure being formally described by 
$\boldsymbol{\tilde{B}}=dec \circ \boldsymbol{B} \circ (enc_{P} \times 
enc_{M})$.} \label{Figure_free_operations_boxes}}
\end{figure} 
\begin{lemma}
\label{lemma:T_is_convex}
Let $\mc{F}$ be the above defined set of free operations. Given 
$B,B^{\prime}$ two 
prepare-and-measure statistics in $\mc{O}$, and given $\pi \in [0,1]$, one 
has:
\beq
\label{lemma:eq_T_is_convex}
T\left(\pi B+\left(1-\pi\right)B^{\prime}\right)=\pi T(B)+\left(1-\pi\right)T(B^{\prime}),
\eeq
where the sum and multiplication on $\mc{O}$ are defined component-wise. 
\end{lemma}

\begin{lemma}
 The set $\mc{F}$ of free operations is closed under composition.
\end{lemma}

It 
remains to show that the transformations belonging to $\mathcal{F}$ really 
fulfill the requirement of being  free operations, \emph{i.e.} we 
must  show that any element $T$ in $\mc{F}$ does not create a 
resource from a resourceless object -- it preserves the set of 
non-contextual prepare-and-measure statistics. More formally: 

\begin{theorem}
 Given a free operation $T \in \mc{F}$, and a \pandms 
$B=\{p(k \vert j,i)\}$, if $B$ admits a universally noncontextual
model, then $\tilde{B}=T(B)$ also admits a universally noncontextual model.
\end{theorem}
\begin{proof}
 Since $B$ admits a universally contextual model (w.r.t. to sets 
$\mc{E}_{M}$ and $\mc{E}_{P}$ of operational equivalences for measurements 
and preparations),  there exist a family of probability spaces 
$\{\Lambda,\Sigma_{i},\mu_{i}\}_{i}$ one for each preparation $P_i$, and a 
set of response functions $\{\xi_{k \vert j}(\lambda)\}$ such that:
\begin{align}
\label{eq:thm_T_preservesNC_noncontextualboxes_begin}
 \forall& \,\, \lambda, j, k: &\xi_{k \vert j}(\lambda) \geq 0;    \\
 \forall& \,\, \lambda,j: &\sum_{k \in K}\xi_{k \vert j}(\lambda)=1; \\
 \forall& \,\, \lambda,r: &\sum_{j \in J}\sum_{k \in 
K}\left(\alpha^{r}_{k \vert j} - \beta^{r}_{k \vert j}\right)\xi_{k \vert 
j}(\lambda)=0; \\
 \forall& \,\, \lambda, i: &\mu_{i}(\lambda) \geq 0; \\
 \forall& \,\, i: &\int_{\Lambda}\mu_{i}(\lambda)=1; \\
 \forall& \,\, \lambda,s: &\sum_{i \in I}\left(\alpha^{s}_{i} - \beta^{s}_{i}\right)\mu_{i}(\lambda)=0; \\
 \forall& \,\, i,j,k: &\int_{\Lambda}\xi_{k \vert 
j}(\lambda)\mu_{i}(\lambda) = p(k \vert j,i).
 \label{eq:thm_T_preservesNC_noncontextualboxes_end}
\end{align}
 Therefore:
 \begin{align}
 &\behavtd{k}{j}{i}: =T(\behav{k}{j}{i}) \\
         &= \Tppm \\
         &=\sum_{k,j,i}q_{O}(\tilde{k} \vert k)\left( \int_{\Lambda}\xi_{k 
\vert 
j}(\lambda)\mu_{i}(\lambda)\right)q_{P}(i 
\vert \td{i} \,\, )q_{M} (j \vert \td{j} \,\,) \\
         &=\int_{\Lambda}\underbrace{\left(\sum_{k,j}q_{O}(\td{k} \vert 
k) \xi_{k \vert j}(\lambda)q_{M}(j \vert \td{j}) 
\right)}_{:=\xi_{\td{k} \vert \td{j}}(\lambda)} \underbrace{\left( 
\sum_{i}\mu_{i}(\lambda)q_{P}(i \vert \td{i})  
\right)}_{:=\mu_{\td{i}}(\lambda)} \\
         &= \int_{\Lambda}\xi_{\td{k} \vert 
\td{j}}(\lambda)\mu_{\td{i}}(\lambda).
 \end{align}
With these new set of response functions $\{\xi_{\td{k} \vert 
\td{j}}(\lambda)\}$ and probability measures $\{\mu_{\td{i}}\}$ over 
$\Lambda$, we are going to prove that $\tilde{B}$ admits a universally 
noncontextual model (w.r.t. the transformed operational equivalences). For 
clarity we break the remaining of the proof into three  statements:
\begin{itemize} 
 \item [i)] The family $\{\xi_{\td{k} \vert 
\td{j}}(\lambda)\}$ constitute an admissible set of response functions.
On one hand:
\begin{align}
\forall \td{k},\td{j},\lambda:& \\
 \xi_{\td{k} \vert 
\td{j}}(\lambda)&:=\sum_{k,j}q_{O}(\td{k} \vert 
k) \xi_{k \vert j}(\lambda)q_{M}(j \vert \td{j}) \\
  &\geq 0.
\end{align}
On the other hand (see 
Fig.~\ref{Figure_free_operations_boxes_rhs_composition}), for all 
$\lambda$ belonging to $\Lambda$:
\begin{align}
 \sum_{\td{k}}&\xi_{\td{k} \vert \td{j}}(\lambda) = 
\sum_{\td{k}}\sum_{k,j}q_{O}(\td{k} \vert 
k) \xi_{k \vert j}(\lambda)q_{M}(j \vert \td{j}) \\
  &=\sum_{\td{k}}\sum_{j} \underbrace{\left[\sum_{k}q_{O}(\td{k}\vert 
k)\xi_{k \vert j}(\lambda) \right]}_{:=q_{\lambda}(\td{k}\vert j)}q_{M}(j 
\vert \td{j}) \\
  &= \sum_{\td{k}} \underbrace{\sum_{j}q_{\lambda}(\td{k}\vert j)q_{M}(j 
\vert \td{j})}_{:=q_{\lambda}(\td{k} \vert \td{j})} = 
\sum_{\td{k}}q_{\lambda}(\td{k} \vert \td{j})=1.
\end{align}

\begin{figure}[]
\includegraphics[scale=0.1]{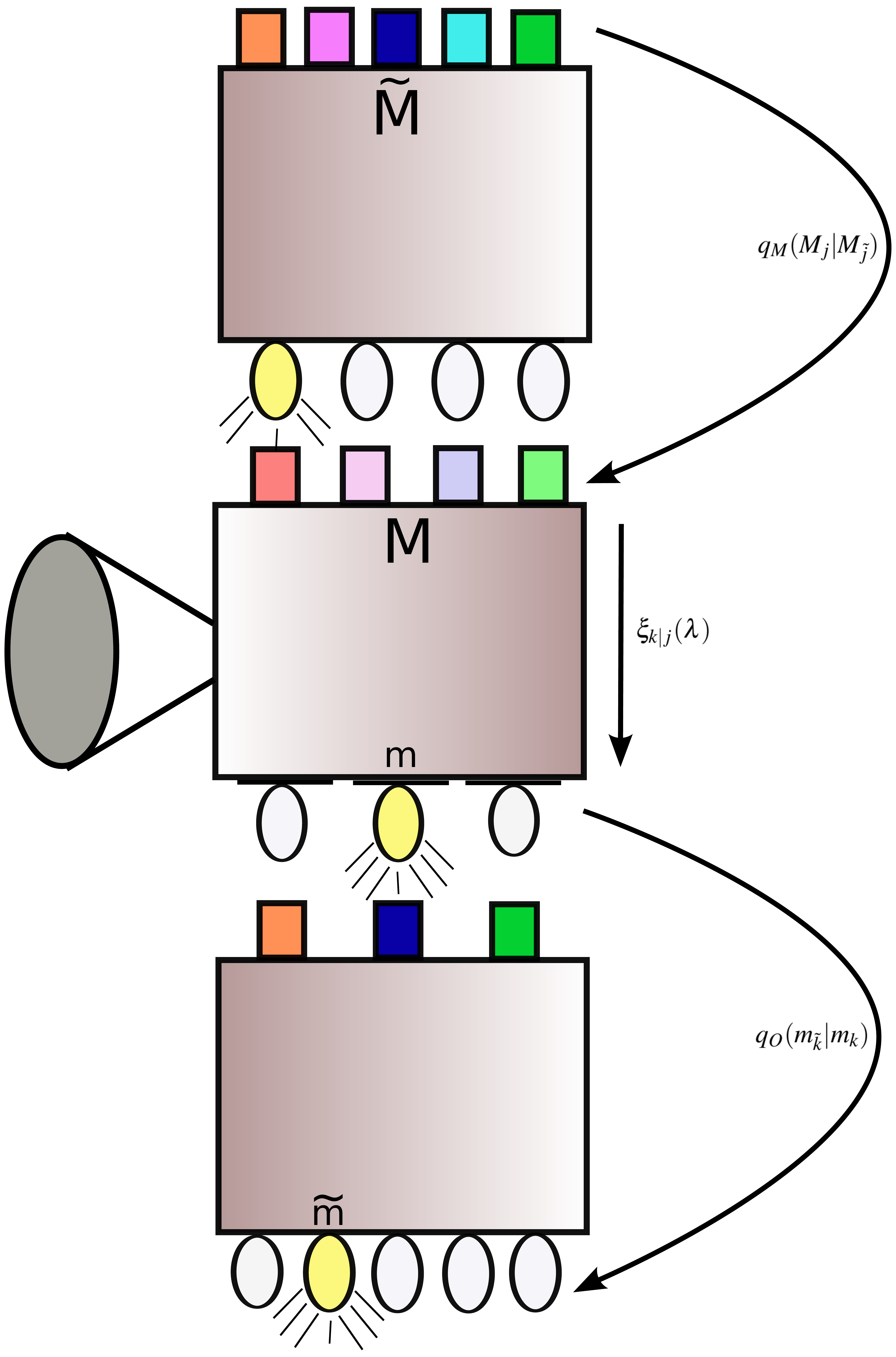}
 \caption{\footnotesize{Schematic drawing of the net effect for the 
composition of a preprocessing on the box $M$ and the post-processing on 
the outcomes. Given every choice of measurement $\tilde{M}_{\td{j}}$ all 
at least one outcome must be obtained. Again, this wiring is nothing but 
a composition of a codification and de-codification on a communication 
channel~\cite{TF17}.} 
\label{Figure_free_operations_boxes_rhs_composition}}
\end{figure}

\item [ii)] The family $\{\mu_{\td{i}}\}$ constitute a set of probability 
measures over $\Lambda$. 

It is clear from the definition that each $\mu_{\td{i}}$ is a non-negative 
function over $\Lambda$. In addition, for each $\td{i} \in \td{I}$, one 
has:
\begin{align}
 \int_{\lambda}\mu_{\td{i}}(\lambda)&=\int_{\Lambda} \sum_{i}q_{P}(i \vert 
\td{i})\mu_{i}(\lambda) \\
           &=\sum_{i}q_{P}(i \vert \td{i}) 
\underbrace{\int_{\lambda}\mu_{i}(\lambda)}_{=1} \\
           &=\sum_{i}q_{P}(i \vert \td{i})  =1.
\end{align}

\item [iii)] The operational equivalences $\mc{E}_{P}$ and $\mc{E}_{M}$ 
are lifted, through the action of free operations in $\mc{F}$, onto other 
sets of operational equivalences $\mc{E}_{\td{P}}$ and $\mc{E}_{\td{M}}$ 
which in turn also respect the principle of preparation noncontextuality 
and measurement noncontextulity~\cite{RS05,SSW17} respectively.

We prove the result for operational equivalences among preparations, 
and the result will follow in complete analogy for operational equivalences
among measurements. Given one operational equivalence among preparations 
for the ordinary prepare-and-measure scenario, labelled by $s$, \emph{i.e.} 
an equality 
\begin{align}
  &\sum_{i \in I}\left(\alpha^{s}_{i} - \beta^{s}_{i}\right)\mu_{i}(\lambda)=0, \forall 
\,\, \lambda \in \Lambda,
\end{align}
we define novel set of coefficients $\{\alpha^{s}_{\td{i}}\}_{\td{i} \in 
\td{I}}$ and $\{\beta^{s}_{\td{i}}\}_{\td{i} \in 
\td{I}}$ satisfying 
\begin{align}
\label{eq:proof_equivalences_td_alpha}
 \forall \,\, i \in [I]: \sum_{\td{i}}\alpha^{s}_{\td{i}}q_{P}(i \vert 
\td{i})=\alpha^{s}_{i},
\end{align}
and
\begin{align}
\label{eq:proof_equivalences_td_beta}
 \forall \,\, i \in [I]: \sum_{\td{i}}\beta^{s}_{\td{i}}q_{P}(i \vert 
\td{i})=\beta^{s}_{i}.
\end{align}
Now, if one defines the new set of operational equivalences for 
preparations $\mc{E}_{\td{P}}$ using 
Eqs.~\eqref{eq:proof_equivalences_td_alpha} 
and~\eqref{eq:proof_equivalences_td_beta}, it is straightforward to check 
that
\begin{equation}
\forall \,\, s, \lambda: 
\sum_{\td{i}}\alpha^{s}_{\td{i}}\mu_{\td{i}}(\lambda)= 
\sum_{\td{i}}\beta^{s}_{\td{i}}\mu_{\td{i}}(\lambda), 
\end{equation}
since these novel operational equivalences come from a lift of operational 
equivalences which are noncontextual in the ordinary (non-)transformed 
scenario.
\end{itemize}
\end{proof}

\subsection{Monotones}
\label{subsec:Quantifiers}

Now that we have defined the sets of objects and free operations, we are in position to verify if the monotones introduced in Refs.
\cite{Barrett05b, Allcock09, GWAN12,
Joshi13, Vicente14, LVN14,  GHHHHJKW14,
GA15, HGJKL15,  GA17, ACTA17, BAC17, ABM17} can be adapted to the framework of generalized contextuality.

\subsubsection{Contextual Fraction}
\label{subsubsec:ConetxtualFract}

The \emph{contexual fraction} is a contextuality quantifier based on the intuitive notion
of what fraction of a given prepare-and-measure statistics admits a noncontextual
description. Formally it  is defined as follows~\cite{AB11, ADLPBC12, ABM17, AT17}:
\begin{align}
\mc{C}: \mc{O} &\longrightarrow [0,1] \\ \nonumber
       B=\{p(k \vert j,i)\}   & \mapsto \mc{C}(B)
\end{align}
where
\begin{align}
1-\mc{C}(B):= \mbox{max } \lambda \nonumber  \\
\qquad \textup{s.t.}\quad &B=\lambda B^{NC}+(1-\lambda)B^{\prime} 
\nonumber\\
& B^{NC} \in \mathsf{NC}\left(\mathcal{S}\right)  \nonumber \\
& B^{\prime}  \in \mc{O}.  \label{def:cf2}
\end{align}

\begin{theorem}
\label{thm:cf_is_monotone}
The contextual-fraction is a resource monotone with respect to $\mc{F}$. 
\end{theorem}
\begin{proof}
 Let $B \in \mc{O}$ be a given prepare-and-measure statistics, and let 
$\mc{C}(B)$ be equal to $\lambda_{\max}$ with the following decomposition:
\begin{equation}
 B=\lambda_{\max}B^{NC}_{\max}+(1-\lambda_{\max})B^{\prime}_{\max}.
\end{equation}
Then:
\begin{align}
T(B)&=T(\lambda_{\max}B^{NC}_{\max}+(1-\lambda_{\max})B^{\prime}_{\max}) \\
&=\lambda_{\max}T(B^{NC}_{\max})+(1-\lambda_{\max})T(B^{\prime}_{\max}) \\
&= \lambda_{\max}B^{NC}+(1-\lambda_{\max})B^{\prime}.
\end{align}
with $B^{NC}=T(B^{NC}_{\max})$ a  noncontextual box, since $T$ preserves the 
noncontextual set, and $B^{\prime}$ a valid object. Then 
$\lambda_{\max(T(B))} \geq \lambda_{\max}$. Therefore:
\beq
1-\mc{C}(T(B)) \geq 1-\mc{C}(B) \Longrightarrow \mc{C}(T(B)) \leq \mc{C}(B).
\eeq
\end{proof}

The contextual fraction  is subadditive under independent juxtapositions.

\begin{theorem}
\label{thm:cf_ad}
 Given two behaviors $B_1$ and $B_2$,we have that 
 \beq \mathcal{C}\left(B_1 \otimes B_2\right) \leq \mathcal{C}\left(B_1 \right)+ 
\mathcal{C}\left(B_2\right)-\mathcal{C}\left(B_1 
\right)\mathcal{C}\left(B_2\right) \leq \mathcal{C}\left(B_1 \right)+ 
\mathcal{C}\left(B_2\right) .\eeq
\end{theorem}

\begin{proof}
 Let $B_1^*$ and $B_2^*$ be the behaviors achieving the minimum in Eq. \eqref{def:cf2} for $B_1$ and $B_2$, respectively,
 with $1-\mathcal{C}\left(B_1 \right)=\lambda_1$ and $1-\mathcal{C}\left(B_2 \right)=\lambda_2$.  The decomposition in 
 Eq. \eqref{def:cf2} implies that 
 \begin{align}
  p\left(k_1 \vert j_1,i_1\right) &\leq \lambda_1  p^*\left(k_1 \vert j_1,i_1\right), \ \forall \ i_1,j_1,k_1,\\
  p\left(k_2 \vert j_2,i_2\right) &\leq \lambda_1  p^*\left(k_2 \vert j_2,i_2\right), \ \forall \ i_2,j_2,k_2,
 \end{align}
 which in turn imply that 
 \beq p\left(k_1 k_2 \vert j_1j_2,i_1i_2\right) \leq \lambda_1\lambda_2 p^*\left(k_1 k_2 \vert j_1j_2,i_1i_2\right), \ \forall \ i_1,j_1,k_1,i_2,j_2,k_2,\label{eq:cf_ad}\eeq
 where $p^*\left(k_1 k_2 \vert j_1j_2,i_1i_2\right)= p^*\left(k_1 \vert j_1,i_1\right)p^*\left(k_2 \vert j_2,i_2\right)$ are the 
 probabilities given by $B_1^* \otimes B_2^*$.
From Eq. \eqref{eq:cf_ad} if follows that there is a behavior $B'$ such that 
\beq B_1 \otimes B_2 = \lambda_1 \lambda_2 B_1^* \otimes B_2^* + B'. \eeq
Hence, we have
\begin{align}
1-\mathcal{C}\left(B_1 \otimes B_2 \right)& \geq \lambda_1 \lambda_2\\
&= \left(1-\mathcal{C}\left(B_1\right)\right)\left(1-\mathcal{C}\left( B_2 \right)\right),
\end{align}
which in turn implies the desired result.
\end{proof}

\subsubsection{Robustness Measures}

Robustness of contextuality is a quantifier based on the intuitive notion of how much noncontextual noise a
given prepare-and-measure statistics can sustain before becoming noncontextual. Given a scenario $\mc{S}$ one defines the \emph{robustness 
measure}~\cite{ACTA17} as follows:

\begin{align}
 \label{def:eq_robustness}
 \mc{R}: \mc{O} &\longrightarrow [0,1] \\
              B  &\mapsto \mc{R}(B) \nonumber
\end{align}
where
\begin{align}
\label{def:eq_minimization_for_robustness}
\mc{R}(B):= \mbox{min } \lambda \nonumber  \\
 \textup{s.t.}\quad &(\lambda B^{NC}+(1-\lambda)B)\in \mathsf{NC}\left(\mathcal{S}\right)  \nonumber\\
                   & B^{NC} \in \mathsf{NC}\left(\mathcal{S}\right) .
\end{align}

\begin{theorem}
\label{thm:robustness_is_monotone}
The robustness measure $\mc{R}$ is a resource monotone with respect to the 
free-operations in $\mc{F}$.
\end{theorem}
\begin{proof}
 Given $B \in \mc{O}$, let $\lambda_{\min}$ be the minimum of the 
optimization problem \eqref{def:eq_minimization_for_robustness} above, in 
which
\beq
\lambda_{\min}B^{NC}_{\min}+(1-\lambda_{\min})B = B^{\ast},
\eeq
with $B^{\ast}$ noncontextual. Then:
\begin{align}
T(B^{\ast})&=T(\lambda_{\min}B^{NC}_{\min}+(1-\lambda_{\min})B) \\
           &=\lambda_{\min}T(B^{NC}_{\min})+(1-\lambda_{\min})T(B) \\
           &=\lambda_{\min}\td{B}^{NC}_{\min}+(1-\lambda_{\min})T(B).
\end{align}
Since, $T$ preserves the noncontextual set, $T(B^{\ast})$ is noncontextual 
as well as $\td{B}^{NC}_{\min}$ and therefore $\lambda_{\min(T(B))}\leq 
\lambda_{\min}$.
\end{proof}

The robustness is subadditive under independent juxtapositions.

\begin{theorem}
 Given two behaviors $B_1$ and $B_2$,we have that 
 \beq \mathcal{R}\left(B_1 \otimes B_2\right) \leq \mathcal{R}\left(B_1 \right)+ 
\mathcal{R}\left(B_2\right)-\mathcal{R}\left(B_1 
\right)\mathcal{R}\left(B_2\right) \leq \mathcal{R}\left(B_1 \right)+ 
\mathcal{R}\left(B_2\right) .\eeq
\end{theorem}
\begin{proof}
The proof follows the same lines of the proof of Thm. \ref{thm:cf_ad}.
\end{proof}

The proof of Thm.\ref{thm:robustness_is_monotone} above suggests that for 
a more restrictive class of free operations, one can relax the assumption 
present at the constraints of the optimization problem 
in.~\eqref{def:eq_minimization_for_robustness}, and instead of optimizing 
over all \pandms which are noncontextual, we could fix a given reference 
noncontextual prepare-and-measure statistics, say $B_{ref}$, and then define the a new 
robustness measure with respect to that fixed prepare-and-measure 
statistics: 

\begin{align}
 \label{def:eq_robustness_ref}
 \mc{R}_{ref}: \mc{O} &\longrightarrow [0,1] \\
              B  &\mapsto \mc{R}_{ref}(B) \nonumber
\end{align}
where
\begin{align}
\label{def:eq_minimization_for_robustness_ref}
\mc{R}_{ref}(B):= \mbox{min } \lambda \nonumber  \\
 \textup{s.t.}\quad &(\lambda B_{ref}+(1-\lambda)B) \in \mathsf{NC}\left(\mathcal{S}\right) .      
\end{align}
This new measure is a monotone whenever $B_{ref}$ is preserved upon action 
of (a more restrictive subset of) free operations. With this on hands it 
is possible to enunciate the following result: 
\begin{corollary}
 Let $\mc{F}_{ref} \subseteq \mc{F}$ be all free operations which preserves
 $B_{ref} \in \mc{O}$, \emph{i.e.} 
\beq
\mc{F}_{ref}:=\{T \in \mc{F}; \,\, T(B_{ref})=B_{ref}\}. 
\eeq
Under this new set of free operations, $\mc{R}_{ref}$ is a resource monotone.
\end{corollary}

The proof follows the same lines as the proof of the proof of Thm.\ref{thm:robustness_is_monotone}.

Notice that the proofs of Thms. \ref{thm:cf_is_monotone} and \ref{thm:robustness_is_monotone} rely only on the fact 
that the operations in $\mathcal{F}$ are linear and preserve 
$\mathsf{NC}\left(\mathcal{S}\right)$. The results in Sec. \ref{sec: GenContextuality} imply that $\mc{C}, \mc{R}$ and $\mc{R}_{ref}$ can be computed
efficiently using linear programming \cite{SSW17}.

\subsection{Kullback-Liebler divergence}

Given two probability distributions $p$ and $q$ in a sample
space $\Omega$, the Kullback-Leiber divergence or relative entropy between 
$p$ and $q$ \beq D_{KL}\left(p\|q\right)= \sum_{i\in \Omega} 
p_i\log\left(\frac{p_i}{q_i}\right)\eeq
is a measure of the difference between the two probability distributions 
$p$ and $q$~\cite{KL51}. With this, one can define the relative 
entropy $D_{KL}\left(B\|B'\right)$ between
two prepare-and-measure statistics $B=\left\{p\left(\cdot \vert j,i\right)\right\}$ and $B'=\left\{p'\left(\cdot \vert j,i\right)\right\}$
as the relative entropy 
between the output distributions obtained from $B$
and $B'$ for the optimal  choice of preparation and measurement:
\beq D_{KL}\left(B\|B'\right) \coloneqq \max_{i,j} \;  D_{KL}\left(p\left(\cdot \vert j,i\right)\|p\left(\cdot \vert j,i\right)\right).\eeq
This quantity measures  the distinguishability of $B$ from $B'$. 
We can now define the relative entropy of contextuality \cite{DGG05, GHHHHJKW14, HGJKL15}
\beq \mc{KL}\left(B\right)\coloneqq \min_{B'\in \mathsf{NC}\left(\mathcal{S}\right)}D_{KL}\left(B\|B'\right), \label{eq:defKL}\eeq
which quantifies the distinguishability of
$B$ from its closest, with respect to $D_{KL}$, noncontextual preapre-and-measure statistics.

\begin{theorem}
 The relative entropy of contextuality $\mathcal{KL}$ is a resource 
monotone with respect to $\mc{F}$.
\end{theorem}

\begin{proof}
 Given $B \in \mathcal{O}$, let $B^*$ be the noncontextual prepare-and-measure statistics achieving the minimum in
 Eq.~\eqref{eq:defKL}. Given $T \in \mathcal{F}$, we have
 \begin{widetext}
 \begin{align}
  \mc{KL}\left(T\left(B\right)\right) &\leq  \max_{\tilde{i},\tilde{j}} D_{KL}\left(p\left(\cdot\vert \tilde{j},\tilde{i}\right) \| p^*\left(\cdot\vert \tilde{j},\tilde{i}\right)\right)\label{eq:KLmon1}\\
  & =  \max_{\tilde{i},\tilde{j}} \sum_{\tilde{k}} p\left(\tilde{k}\vert \tilde{j},\tilde{i}\right) \log\left[\frac{p\left(\tilde{k}\vert \tilde{j},\tilde{i}\right)}{p^*\left(\tilde{k}\vert \tilde{j},\tilde{i}\right)}\right]\label{eq:KLmon2}\\
& = \max_{\tilde{i},\tilde{j}} \sum_{\tilde{k}} \sum_{i,j,k}q_O\left(\tilde{k}\vert k\right)p\left(k\vert j,i\right) q_M\left(j\vert \tilde{j}\right)q_P\left(i\vert \tilde{i}\right)\log\left[\frac{\sum_{i,j,k}q_O\left(\tilde{k}\vert k\right)p\left(k\vert j,i\right) q_M\left(j\vert \tilde{j}\right)q_P\left(i\vert \tilde{i}\right)}{\sum_{i,j,k}q_O\left(\tilde{k}\vert k\right)p^*\left(k\vert j,i\right) q_M\left(j\vert \tilde{j}\right)q_P\left(i\vert \tilde{i}\right)}\right]\label{eq:KLmon3}\\
& \leq  \max_{\tilde{i},\tilde{j}} \sum_{\tilde{k}} \sum_{i,j,k} q_O\left(\tilde{k}\vert k\right)p\left(k\vert j,i\right) q_M\left(j\vert \tilde{j}\right)q_P\left(i\vert \tilde{i}\right)\log\left[\frac{p\left(k\vert j,i\right) }{p^*\left(k\vert j,i\right) }\right]\label{eq:KLmon4}\\
&=  \max_{\tilde{i},\tilde{j}}  \sum_{i,j,k} p\left(k\vert j,i\right) q_M\left(j\vert \tilde{j}\right)q_P\left(i\vert \tilde{i}\right)\log\left[\frac{p\left(k\vert j,i\right) }{p^*\left(k\vert j,i\right) }\right]\label{eq:KLmon5}\\
&= \max_{\tilde{i},\tilde{j}}  \sum_{i,j}q_M\left(j\vert \tilde{j}\right)q_P\left(i\vert \tilde{i}\right)  \left(\sum_k p\left(k\vert j,i\right) \log\left[\frac{p\left(k\vert j,i\right) }{p^*\left(k\vert j,i\right) }\right]\right)\label{eq:KLmon6}\\
&\leq \max_{i,j} \sum_k p\left(k\vert j,i\right) \log\left[\frac{p\left(k\vert j,i\right) }{p^*\left(k\vert j,i\right) }\right]\label{eq:KLmon7}\\
&=\mathcal{KL}\left(B\right).
\end{align}
\end{widetext}
Eq.~\eqref{eq:KLmon1} follows form the definition of $\mathcal{KL}\left(T(B)\right)$ and the fact that $T\left(B^*\right) \in \mathsf{NC}\left(\mathcal{S}\right)$,
Eq.~\eqref{eq:KLmon2} follows from the definition of $D_{KL}$, Eq.~\eqref{eq:KLmon3} follows from the definition of $T\in \mathcal{F}$,
Eq.~\eqref{eq:KLmon4} follows from the log sum inequality, Eq.~\eqref{eq:KLmon5} follows from basic algebra, Eq.~\eqref{eq:KLmon6} follows from
$\sum_{\tilde{k}} q_{O}\left(\tilde{k}\vert k\right)=1$ and  Eq.~\eqref{eq:KLmon7} follows from the fact that 
the average is smaller than the
largest value.
\end{proof}

The relative entropy of contextuality is subadditive under independent juxtapositions.

\begin{theorem}
 Given two behaviors $B_1$ and $B_2$,we have that 
 \beq \mathcal{KL}\left(B_1 \otimes B_2\right) \leq \mathcal{KL}\left(B_1 \right)+ \mathcal{KL}\left(B_2\right).\eeq
\end{theorem}

\begin{proof}
 Let $B_1^*$ and $B_2^*$ be  behaviors achieving the minimum in Eq. \eqref{eq:defKL} for $B_1$ and $B_2$, respectively.
 Then, we have
 \begin{align}
  \mathcal{KL}\left(B_1 \otimes B_2\right)&\leq D_{KL}\left(B_1\otimes B_2\|B_1^*\otimes B_2^*\right) \label{eq:kl_ad1}\\
  &= \max_{i_1,i_2,j_1,j_2} \sum_{k_1,k_2} p\left(k_1k_2\vert j_1j_2, i_1i_2\right)\log\left(\frac{p\left(k_1k_2\vert j_1j_2, i_1i_2\right)}{p^*\left(k_1k_2\vert j_1j_2, i_1i_2\right)}\right)\label{eq:kl_ad2}\\
  &=\max_{i_1,i_2,j_1,j_2} \sum_{k_1,k_2} p\left(k_1\vert j_1, i_1\right)p\left(k_2\vert j_2, i_2\right)\log\left(\frac{p\left(k_1\vert j_1, i_1\right)p\left(k_2\vert j_2, i_2\right)}{p^*\left(k_1\vert j_1, i_1\right)p^*\left(k_2\vert j_2, i_2\right)}\right)\label{eq:kl_ad3}\\
 &= \max_{i_1,i_2,j_1,j_2} \left[\sum_{k_1}  p\left(k_1\vert j_1, 
i_1\right)\log\left(\frac{ p\left(k_1\vert j_1, i_1\right)}{ p^*\left(k_1\vert j_1, 
i_1\right)}\right) + \sum_{k_2} p\left(k_2\vert j_2, 
i_2\right)\log\left(\frac{p\left(k_2\vert j_2, i_2\right)}{p^*\left(k_2\vert j_2, 
i_2\right)}\right)\right]\label{eq:kl_ad4}\\
 & \leq \max_{i_1,j_1} \sum_{k_1}  p\left(k_1\vert j_1, i_1\right)\log\left(\frac{ 
p\left(k_1\vert j_1, i_1\right)}{ p^*\left(k_1\vert j_1, i_1\right)}\right) + 
\max_{i_2,j_2}\sum_{k_2} p\left(k_2\vert j_2, i_2\right)\log\left(\frac{p\left(k_2\vert 
j_2, i_2\right)}{p^*\left(k_2\vert j_2, i_2\right)}\right)\label{eq:kl_ad5}\\
 &=\mathcal{KL}\left(B_1\right)+\mathcal{KL}\left(B_2\right). \label{eq:kl_ad6}
 \end{align}
Eq. \eqref{eq:kl_ad1} follows from the definition of $\mathcal{KL}\left(B_1 \otimes B_2\right)$, Eq. \eqref{eq:kl_ad2} follows from the definition
of $D_{KL}$ and $B_1\otimes B_2$, Eq. \eqref{eq:kl_ad3} follows from additivity of the relative entropy for independent distributions,
Eq. \eqref{eq:kl_ad4} follows from basic algebra, and Eq. \eqref{eq:kl_ad5} follows from  the fact that 
$B_1^*$ and $B_2^*$ are  behaviors achieving the minimum in Eq. \eqref{eq:defKL} for $B_1$ and $B_2$, respectively.

\end{proof}

\subsection{Distance based monotones}
We now introduce contextuality monotones based on geometric distances, in contrast with the previous defined
quantifier which is based on entropic distances, replacing the relative entropy by some geometric distance defined
over real vector spaces in  Eq.~\eqref{eq:defKL} \cite{BAC17,AT17}. 
Let $D$ be any distance defined in real vector spaces $\mathds{R}^{K}$. We 
define the distance between two 
prepare-and-measure statistics $B$ and $B'$ as
\beq D\left(B, B'\right) \coloneqq \max_{i,j} \;  D\left(p\left(\cdot \vert j,i\right),p\left(\cdot \vert j,i\right)\right).\label{eq:defD_box}\eeq
We can now define the $D$-contextuality distance
\beq \mc{D}\left(B\right)\coloneqq \min_{B'\in \mathsf{NC}\left(\mathcal{S}\right)}D\left(B,B'\right), \label{eq:defD}\eeq
which quantifies the distance, with respect to $D$, from $B$ to the set of noncontextual prepare-and-measure statistics. We focus here on the contextuality quantifier 
obtained when we use the $\ell_1$ norm
\beq D_1\left(x,y\right)=\sum_i \left|x_i-y_i\right|.\eeq

\begin{theorem}
 The $\ell_1$-contextuality distance is a resource monotone with respect 
to the free-operations in $\mc{F}$.
\end{theorem}

\begin{proof}
  Given $B \in \mathcal{O}$, let $B^*$ be the noncontextual prepare-and-measure statistics achieving the minimum in
 Eq.~\eqref{eq:defD}. Given $T \in \mathcal{F}$, we have
 \begin{widetext}
  \begin{align}
   \mathcal{D}\left(T\left(B\right)\right)&\leq \max_{\tilde{i},\tilde{j}}\sum_{\tilde{k}} \left|p\left(\tilde{k}\vert \tilde{j},\tilde{i}\right)-p^*\left(\tilde{k}\vert \tilde{j}, \tilde{i}\right)\right| \label{eq:dist1}\\   
  &= \max_{\tilde{i},\tilde{j}} \sum_{\tilde{k}} \left| \sum_{i,j,k} q_O\left(\tilde{k}\vert k\right)\left(p\left(k\vert j,i\right)-p^*\left(k\vert j,i\right)\right)q_M\left(j\vert \tilde{j}\right)q_P\left(i\vert \tilde{i}\right)\right|\label{eq:dist2}\\
  &\leq \max_{\tilde{i},\tilde{j}} \sum_{\tilde{k}} \sum_{i,j,k} q_O\left(\tilde{k}\vert k\right)q_M\left(j\vert \tilde{j}\right)q_P\left(i\vert \tilde{i}\right)\left| \left(p\left(k\vert j,i\right)-p^*\left(k\vert j,i\right)\right)\right|\label{eq:dist3}\\
  &= \max_{\tilde{i},\tilde{j}} \sum_{i,j,k} q_M\left(j\vert \tilde{j}\right)q_P\left(i\vert \tilde{i}\right)\left| \left(p\left(k\vert j,i\right)-p^*\left(k\vert j,i\right)\right)\right|\label{eq:dist4}\\
  &= \max_{i,j} \sum_k \left|p\left(k\vert j,i\right) - p^*\left(k\vert j,i\right)\right|\label{eq:dist5}\\
  &=\mathcal{D}\left(B\right).
  \end{align}
 \end{widetext}
 Eq.~\eqref{eq:dist1} follows form the definition of $\mathcal{D}\left(T(B)\right)$ and the fact that
 $T\left(B^*\right) \in \mathsf{NC}\left(\mathcal{S}\right)$,
Eq.~\eqref{eq:dist2} follows from the definition of $T\in \mathcal{F}$,
Eq.~\eqref{eq:dist3} follows from the triangular inequality for the $\ell_1$ norm, Eq.~\eqref{eq:dist4} follows from
$\sum_{\tilde{k}} q_{O}\left(\tilde{k}\vert k\right)=1$ and  Eq.~\eqref{eq:dist5} follows from the fact that 
the average is smaller than the
largest value. \end{proof}

The $\ell_1$-contextuality distance is subadditive under independent juxtapositions.

\begin{theorem}
\label{thm:D_add}
 Given two behaviors $B_1$ and $B_2$,we have that 
 \beq \mathcal{D}\left(B_1 \otimes B_2\right) \leq \mathcal{D}\left(B_1 \right)+ \mathcal{D}\left(B_2\right).\eeq
\end{theorem}

\begin{proof}
 Let $B_1^*$ and $B_2^*$ be the behaviors achieving the minimum in Eq. \eqref{eq:defD} for $B_1$ and $B_2$, respectively.
 Then, we have
 \begin{align}
  \mathcal{D}\left(B_1 \otimes B_2\right)&\leq D\left(B_1\otimes B_2,B_1^*\otimes B_2^*\right) \label{eq:D_ad1}\\
  &= \max_{i_1,i_2,j_1,j_2} \sum_{k_1,k_2} \left|p\left(k_1k_2\vert j_1j_2, i_1i_2\right)-p^*\left(k_1k_2\vert j_1j_2, i_1i_2\right)\right|\label{eq:D_ad2}\\
  &=\max_{i_1,i_2,j_1,j_2} \sum_{k_1,k_2} \left|p\left(k_1\vert j_1, i_1\right)p\left(k_2\vert j_2, i_2\right)-p^*\left(k_1\vert j_1, i_1\right)p^*\left(k_2\vert j_2, i_2\right)\right|\label{eq:D_ad3}\\
   &\leq \max_{i_1,i_2,j_1,j_2} \sum_{k_1,k_2} \left[\left|p\left(k_1\vert j_1, i_1\right)p\left(k_2\vert j_2, i_2\right)-p^*\left(k_1\vert j_1, i_1\right)p\left(k_2\vert j_2, i_2\right)\right|\right.\nonumber \\
   & \hspace{8em} + \left.\left|p^*\left(k_1\vert j_1, i_1\right)p\left(k_2\vert j_2, i_2\right)-p^*\left(k_1\vert j_1, i_1\right)p^*\left(k_2\vert j_2, i_2\right)\right|\right]\label{eq:D_ad4}\\
  &\leq \max_{i_1,i_2,j_1,j_2} \left[\sum_{k_1, k_2} p\left(k_2\vert j_2, i_2\right)\left|p\left(k_1\vert j_1, i_1\right)-p^*\left(k_1\vert j_1, i_1\right)\right|\right.\nonumber \\
    &  \hspace{8em} + \left.\sum_{k_1,k_2}p^*\left(k_1\vert j_1, i_1\right)\left|p\left(k_2\vert j_2, i_2\right)-p^*\left(k_2\vert j_2, i_2\right)\right|\right]\label{eq:D_ad5} \\
    &= \max_{i_1,i_2,j_1,j_2} \left[\sum_{k_1}\left|p\left(k_1\vert j_1, i_1\right)-p^*\left(k_1\vert j_1, i_1\right)\right|
   + \sum_{k_2}\left|p\left(k_2\vert j_2, i_2\right)-p^*\left(k_2\vert j_2, i_2\right)\right|\right]\label{eq:D_ad6} \\
   &=\mathcal{D}\left(B_1 \right)+ \mathcal{D}\left(B_2\right)
 \end{align}
Eq. \eqref{eq:kl_ad1} follows from the definition of $\mathcal{KL}\left(B_1 \otimes B_2\right)$, Eq. \eqref{eq:kl_ad2} follows from the definition
of $D_{KL}$ and $B_1\otimes B_2$, Eq. \eqref{eq:kl_ad3} follows from additivity of the relative entropy for independent distributions,
Eq. \eqref{eq:kl_ad4} follows from basic algebra, and Eq. \eqref{eq:kl_ad5} follows from the definition of the fact that 
$B_1^*$ and $B_2^*$ be the behaviors achieving the minimum in Eq. \eqref{eq:defKL} for $B_1$ and $B_2$, respectively.

\end{proof}

\subsection{Trace Distance}

Instead of taking the maximum over preparations and measurements in Eq.~\eqref{eq:defD_box}, we can take the average value 
to define the uniform $D$-contextuality distance
%
\begin{equation}
 \label{eq:def_trace_distance}
 \mathcal{D}_u(B):= \frac{1}{2 I  J}\min_{B^{\prime} \in \mathsf{NC}\left(\mc{S}\right)} 
D\left( B,B^{\prime} \right) 
\end{equation}
with $B^{\prime}$ taken over all noncontextual prepare-and-measure 
statistics, and $D$ being some distance defined over real vector spaces. Of special importance is the uniform contextuality distance defined
by the trace 
norm $\ell_1$~\cite{NC00,BAC17, AT17}. 



Again, the trace distance $\mathcal{D}_u$ is subadditive under independent 
juxtapositions.

\begin{theorem}
 Given two behaviors $B_1$ and $B_2$,we have that 
 \beq \mathcal{D}_u\left(B_1 \otimes B_2\right) \leq \mathcal{D}_u\left(B_1 \right)+ \mathcal{D}_u\left(B_2\right).\eeq
\end{theorem}
\begin{proof}
The proof is analogous to the proof of Thm. \ref{thm:D_add}.
\end{proof}

Although $\mathcal{D}_u$ is not a monotone under the entire class of of free operations $\mathcal{F}$,
it is a suitable contextuality quantifier when the sets of preparations and measurements are fixed, with the advantage that, unlike
$\mathcal{L}$ and $\mathcal{D}$, the uniform contextuality distance $\mathcal{D}_u$ defined with the $\ell_1$ norm can be computed
efficiently using linear programming.
\section{Conclusion}
\label{sec:conc}

Motivated by the recognition of contextuality as a potential resource for 
computation and information processing, we develop a resource theory for generalized contextuality
that can be applied to arbitrary prepare-and-measure experiments. We 
introduce a minimal  
set of free operations (minimal in the sense that any other possible, and 
physically meaningful, set of physical operations should contain ours as 
a subset) with a clear
operational interpretation and  explicit analytical parametrization and show that several natural contextuality quantifiers
are indeed monotones under this class of free operations. With the recognition that membership testing in the set  of
noncontextual prepare-and-measure statistics can be done efficiently using linear programming, many of these quantifiers can also be computed 
efficiently in the same way.
This framework is useful to classify, quantify,
and manipulate contextuality as a formal resource.
It would be interesting to investigate whether there is a maximally-contextual single prepare-and-measure statistics that 
serve as contextuality bits  for all scenarios, or to identify
what is the simplest scenario admitting inequivalent (not freely
interconvertible) classes of contextuality. Another important issue is to investigate  protocols for contextuality distillation
relying only on the set of free operations. This framework provides a
new interpretation of generalized contextuality, now considered a useful resource rather than an odd feature 
exhibited for quantum physics~\cite{CSW10,CSW14,ATC14}. Indeed, as it has
occurred with entanglement~\cite{BG15,HHHH09,RK17,SHN16} over the years, 
we expect that works like the present one can shed new light on the 
phenomenon, giving to it new insights and making it easier to understand.

\begin{acknowledgments}
The authors thank Marcelo Terra Cunha for 
all suggestions on the manuscript, and the International Institute 
of Physics for its support and hospitality.
BA and CD also acknowledges financial support from the Brazilian 
ministries MEC and MCTIC, and CNPq.
\end{acknowledgments}

\bibliography{biblio2}

\end{document}